\documentclass[submission,copyright,creativecommons]{eptcs}
\providecommand{\event}{LFMTP 2025}

\usepackage{iftex}

\ifpdf
  \usepackage{underscore}         % Only needed if you use pdflatex.
  \usepackage[T1]{fontenc}        % Recommended with pdflatex
\else
  \usepackage{breakurl}           % Not needed if you use pdflatex only.
\fi

\usepackage{inconsolata}
\usepackage{mathtools} 
\usepackage{amsfonts}
\usepackage{amsthm} 
\usepackage{amssymb} 
\usepackage{newunicodechar} 
\usepackage{hyperref} 
\usepackage{bussproofs}
\usepackage{fancyvrb}
\usepackage{graphicx}
\usepackage{thmtools}
\usepackage{color}
\usepackage[capitalize]{cleveref}
\usepackage{hyperref}
\usepackage[inline]{enumitem}
\usepackage{caption}
\usepackage{subcaption}

\title{On the Formal Metatheory of the Pure Type Systems using One-sorted Variable Names and Multiple Substitutions\thanks{This work is partly supported by Agencia Nacional de Investigaci\'on e Innovaci\'on (ANII), Uruguay}}
\author{Sebasti\'an Urciuoli
\institute{Universidad ORT Uruguay, Uruguay}
\email{urciuoli@ort.edu.uy}
\def\titlerunning{Formal Metatheory of Lambda Calculi}
\def\authorrunning{S. Urciuoli}
}

\def\titlerunning{Formal Metatheory of the PTS using Multiple Substitutions} 
\def\authorrunning{S. Urciuoli}

\newtheorem{theorem}{Theorem}[section]
\newtheorem{corollary}{Corollary}[theorem]
\newtheorem{lemma}[theorem]{Lemma}

\newlist{lemenum}{enumerate}{1} 
\setlist[lemenum]{label=(\roman*),ref=\thelemma.(\roman*),noitemsep}

\newlist{lemenum*}{enumerate*}{1} 
\setlist[lemenum*]{label=(\roman*),ref=\thelemma.(\roman*),noitemsep,itemjoin={{; }}}

\newcommand{\betaconv}{\ensuremath{\ {\simeq}\beta}\ }

\newcommand{\ok}[1]{\ensuremath{#1\,\mathsf{ok}}}
\newcommand{\fv}[1]{\ensuremath{\mathsf{fv}\,#1}}
\newcommand{\dom}[1]{\ensuremath{\mathsf{dom}\,#1}}

\newcommand{\free}[2]{\ensuremath{#1 * #2}}
\newcommand{\fresh}[2]{\ensuremath{#1\,{\#}\,#2}}

\newcommand{\freshrg}{\ensuremath{\,{\#}{\downharpoonright}\,}}

\newcommand{\Var}{\ensuremath{\mathcal{V}}}
\newcommand{\Const}{\ensuremath{\mathcal{C}}}

\newcommand{\alphaconv}{\ensuremath{\ {\sim}{\alpha}\ }}
\newcommand{\alphaconvg}{\ensuremath{{\sim}{\alpha}}}
\newcommand{\typerule}[1]{{$\vdash$#1}}
\newcommand{\appendg}{\ensuremath{\,{+}\!{+}\,}}
\newcommand{\stdpts}{\ensuremath{\vdash_\mathsf{s}}}
\newcommand{\stdok}[1]{\ensuremath{#1\,\mathsf{ok}_\mathsf{s}}}

\newunicodechar{ℕ}{\ensuremath{\mathbb{N}}}
\newunicodechar{𝒱}{\Var}
\newunicodechar{𝒞}{\Const}
\newunicodechar{𝕧}{\ensuremath{{}_\mathbb{V}}}
\newunicodechar{≟}{\ensuremath{\stackrel{?}{=}}}
\newunicodechar{≡}{\ensuremath{\equiv}}
\newunicodechar{∀}{\ensuremath{\mathnormal\forall}}
\newunicodechar{∃}{\ensuremath{\exists}}
\newunicodechar{𝒰}{\ensuremath{\mathcal{U}}}
\newunicodechar{𝒞}{\ensuremath{\mathcal{C}}}
\newunicodechar{ℛ}{\ensuremath{\mathcal{R}}}
\newunicodechar{𝒮}{\ensuremath{\mathcal{S}}}
\newunicodechar{α}{\ensuremath{\alpha}}
\newunicodechar{β}{\ensuremath{\beta}}
\newunicodechar{ι}{\ensuremath{\iota}}
\newunicodechar{σ}{\ensuremath{\mathsf{\sigma}}}
\newunicodechar{ρ}{\ensuremath{\mathsf{\rho}}}
\newunicodechar{λ}{\ensuremath{\mathsf\lambda}}
\newunicodechar{Λ}{\ensuremath{\mathsf\Lambda}}
\newunicodechar{Γ}{\ensuremath{\mathsf\Gamma}}
\newunicodechar{Δ}{\ensuremath{\Delta}}
\newunicodechar{Π}{\ensuremath{\mathsf{\Pi}}}
\newunicodechar{Σ}{\ensuremath{\Sigma}}
\newunicodechar{Χ}{X}
\newunicodechar{χ}{\ensuremath{\chi}}
\newunicodechar{∉}{\ensuremath{\not\in}}
\newunicodechar{∈}{\ensuremath{\in}}
\newunicodechar{⊑}{\ensuremath{\sqsubseteq}}
\newunicodechar{∼}{\ensuremath{\sim}}
\newunicodechar{≈}{\ensuremath{\approx}}
\newunicodechar{≃}{\ensuremath{\simeq}}
\newunicodechar{≢}{\ensuremath{\not\equiv}}
\newunicodechar{⊢}{\ensuremath{\vdash}}
\newunicodechar{⊆}{\ensuremath{\subseteq}}
\newunicodechar{≼}{\ensuremath{\preceq}}
\newunicodechar{∪}{\ensuremath{\cup}}
\newunicodechar{≅}{\ensuremath{\cong}}
\newunicodechar{₀}{\ensuremath{{}_\mathtt{0}}}            
\newunicodechar{₁}{\ensuremath{{}_\mathtt{1}}}
\newunicodechar{₂}{\ensuremath{{}_\mathtt{2}}}
\newunicodechar{₃}{\ensuremath{{}_\mathtt{3}}}
\newunicodechar{ₗ}{\ensuremath{{}_\mathtt{l}}}
\newunicodechar{ₜ}{\ensuremath{{}_\mathtt{t}}}
\newunicodechar{ᴸ}{\ensuremath{{}^\mathtt{L}}}
\newunicodechar{ᶠ}{\ensuremath{{}^\mathtt{f}}}
\newunicodechar{ₙ}{\ensuremath{{}_\mathtt{n}}}
\newunicodechar{ₛ}{\ensuremath{{}_\mathtt{s}}}
\newunicodechar{ₐ}{\ensuremath{{}_\mathtt{a}}}
\newunicodechar{⊥}{\ensuremath{\bot}}
\newunicodechar{∶}{{\ensuremath{:}}}
\newunicodechar{∷}{\ensuremath{\dblcolon}}
\newunicodechar{;}{\ensuremath{;}}
\newunicodechar{↔}{\ensuremath{\leftrightarrow}}
\newunicodechar{⇀}{\ensuremath{\rightharpoonup}}
\newunicodechar{⇁}{\ensuremath{\rightharpoondown}}
\newunicodechar{←}{\ensuremath{\leftarrow}}
\newunicodechar{→}{\ensuremath{\rightarrow}}
\newunicodechar{⇒}{\ensuremath{\Rightarrow}}
\newunicodechar{⇔}{\ensuremath{\Leftrightarrow}}
\newunicodechar{⇉}{\ensuremath{\rightrightarrows}}
\newunicodechar{⊎}{\ensuremath{\uplus}}
\newunicodechar{♯}{$\sharp$}
\newunicodechar{⇂}{\ensuremath{\downharpoonright}}
\newunicodechar{▹}{\ensuremath{\triangleright}}
\newunicodechar{≺}{\ensuremath{{\prec}}}
\newunicodechar{∙}{\ensuremath{\bullet}}
\newunicodechar{∘}{\ensuremath{\circ}}
\newunicodechar{◎}{\ensuremath{\circledcirc}}
\newunicodechar{⊙}{\ensuremath{\odot}}
\newunicodechar{·}{\ensuremath{\cdot}}
\newunicodechar{′}{'}

\fvset{fontsize=\small,codes={\catcode`$=3}}

\begin{document}

\maketitle            

\begin{abstract}
We develop formal theories of conversion for Church-style lambda-terms with Pi-types in first-order syntax using one-sorted variables names and Stoughton's multiple substitutions.
We then formalize the Pure Type Systems along some fundamental metatheoretic properties: weakening, syntactic validity, closure under alpha-conversion and substitution.
Finally, we compare our formalization with others related.
The whole development has been machine-checked using the Agda system. 
Our work demonstrates that the mechanization of dependent type theory by using conventional syntax and without identifying alpha-convertible lambda-terms is feasible. 
\end{abstract}

\section{Introduction}

In \cite{tasistro2015}, Tasistro et al. developed a framework in Agda \cite{norell2007} 
containing Stoughton's theories of substitution and $\alpha$-conversion \cite{stoughton88} for the pure 
$\lambda$-calculus in its conventional syntax, i.e., first-order with only one sort of names to serve for both free and bound variables.
Their prime motivation was to test whether such concrete approach was in any way amenable to full formalization.

The use of Stoughton's multiple substitutions brings about the possibility to define a capture-avoiding and structurally recursive substitution operation by renaming the bound variables at the same time the operation is taking place. 
In contrast, the classical formal definition for the unary substitution, 
e.g., by Curry-Feys \cite{curry1958} and Hindley-Seldin \cite{hindley1986}, 
is non-structural because, in the case of $\lambda$-abstractions, 
it invokes itself twice, once to rename the bound name, and again to perform the actual substitution; the latter is on a $\lambda$-term that it is not a proper component of the input.
Non-structural recursion is not ideal in mechanizations because, in general, they require one to conduct many proofs by complete induction on the length of the syntax. 
In addition to being structurally recursive, Stoughton's substitution operation actually renames all bound variables without exempting those who may actually do not cause trouble. Perhaps counterintuitively, this turns out to be very welcome because many proofs that follow the structure of the syntax can be carried out smoothly without having to scrutinize the name of the bound variables in the case of $\lambda$-abstractions, leading to a reduction in the number of cases considered. 

The framework has been put to the test since then to verify many results about  the pure $\lambda$-calculus and the simply-typed $\lambda$-calculus (STLC). In \cite{copello2016}, the authors formalized the Church-Rosser (CR) theorem for the pure $\lambda$-calculus with $\beta$-conversion following the method by Martin-Löf and Tait and subject reduction for the STLC. In \cite{copes2018}, the authors mechanized a proof of the standardization theorem by Kashima \cite{kashima}. In \cite{urciuoli2020}, a formal proof of the strong normalization theorem for STLC by Joachimski and Matthes \cite{joachimski} was presented. Finally, in \cite{urciuoli2023} a proof of strong normalization for System T by Girard \cite{girard} was mechanically verified. 
In spite of having to consider $\alpha$-conversion explicitly, the reports revealed that the workload in terms of labour was still manageable.

Now in this work we wish to test whether our approach scales well to more complex languages and so we formalize a $\lambda$-calculus with dependent types: the Pure Type Systems (PTS) \cite{barendregt92}.
The PTS is a generalization of the $\lambda$-cube that represent many systems  under the same formalism and so it allows to study many properties for all of them. 
 
In order to mechanize the methatheory of the PTS, first we have to extend the syntax of the framework with type-annotated or Church-style $\lambda$-abstractions and $\Pi$-types. As a consequence, we have to revisit an important aspect of the theory of substitutions, namely that of restrictions, i.e., the confinement of their domains, due to some technical reason that we will address in due course. Also, we generalize the type of variables so it can be any in one-to-one correspondence with the natural numbers; in previous works, names were identified with numbers, which takes away a bit of the fun of using names. 
Besides, with this generalization we allow to use strings for the variables thus narrowing the gap between a potential practical implementation of a type-checker (using names) and its certification. 
Finally, we give a more accurate definition of $\alpha$-conversion which will enable us to prove a key result in the theory of $\alpha$-conversion using simpler methods than in previous work. 
All four previous features compelled us to virtually rewrite the framework entirely. 
The result is an Agda library with theories of substitution and $\alpha$-conversion for the underlying language of the PTS: a Church-style $\lambda$-calculus with $\Pi$-types. 

Once we have the new framework we formalize 
some fundamental metatheoric properties of the PTS, including: thinning (weakening), syntactic validity, closure under $\alpha$-conversion and substitution.
Following \cite{mckinna99}, we define the type system by using generalized induction, a technique which will enable us to prove thinning by only means of structural induction. Later we will show that such presentation is extensionally equivalent to a more standard one that does not feature such rule schemes.
Besides, in the course of the development of the metatheory we discuss a problem we had to face arising from the use of renaming substitutions that it seems it only appears in the context of dependent types and which has not been commented anywhere yet to the best of our knowledge. 

The structure of this work is as follows. In the next section we introduce the new framework. Some results are adapted from previous work and some others are new. From now on, all results constitute new development. In \cref{sec:calculus} we present the PTS along with some basic properties. In \cref{sec:metatheory} we formalize thinning, syntactic validity, closure under $\alpha$-conversion and substitution. In \cref{sec:related} we comment on related work. Finally, in \cref{sec:conclusions} we give some code metrics and conclude. 

Throughout this paper we will use Agda notation for definitions and lemmas, and a mix with English for the proofs with the hope of making reading more enjoyable. Some background on any dependently-typed programming is preferable, though it might be enough for some readers to have an understanding in functional programming, e.g., Haskell. The complete sources can be found at the following link: \url{https://github.com/surciuoli/pts-metatheory}.

\section{The Framework}

In this section we define the syntax of the $\lambda$-terms, substitution and $\alpha$-conversion based on the work by Stoughton \cite{stoughton88} (extended to the language of our interest). Also, we prove some important results that we will need in the following sections.

\subsection{Syntax}

Let \Var, the \textit{variables}, be any type such that, first, propositional equality is decidable, 
and second, there are two functions \Verb|encode| and \Verb|decode| such that \Verb|encode| maps every variable to a number and \Verb|decode| is its right inverse. 
We shall use meta-variables $x$, $y$, etc. to range them. Let \Const, the \textit{constants}, be any type and whose elements are ranged by $k$ and $s$. The abstract syntax of the \textit{$\lambda$-terms} ($\Lambda$) is defined inductively in \cref{fig:syntax}.

\begin{figure}
\centering
\begin{subfigure}[b]{0.4\textwidth}
\begin{Verbatim}
data Λ : Set where
  c : 𝒞 → Λ    
  v : 𝒱 → Λ
  λ[_∶_]_ : 𝒱 → Λ → Λ → Λ
  Π[_∶_]_ : 𝒱 → Λ → Λ → Λ    
  _·_ : Λ → Λ → Λ
\end{Verbatim}
\caption{Syntax}
\label{fig:syntax}
\end{subfigure}
\begin{subfigure}[b]{0.4\textwidth}
\begin{Verbatim}
fv : Λ → List 𝒱
fv (c _) = []
fv (v x) = [ x ]
fv (λ[ x ∶ A ] M) = fv A ++ (fv M - x)
fv (Π[ x ∶ A ] B) = fv A ++ (fv B - x)
fv (M · N) = fv M ++ fv N
\end{Verbatim}
\caption{Free Variables}
\label{fig:fv}
\end{subfigure}
\caption{$\lambda$-terms}
\end{figure}

In \cref{fig:fv} we define the function that returns the free names in a given $\lambda$-term, where \Verb|_++_| is list concatenation and \Verb|_-_ : List 𝒱 → 𝒱 → List 𝒱| the function that deletes every occurrence of some name in a given list.
We can then base ourselves upon \Verb|fv| to define whether a name occurs free in some $\lambda$-term or not by: 
\Verb|x $*$ M = x ∈ fv M| and read $x$ is free in $M$ and oppositely, \Verb|x # M = x ∉ fv M| and read $x$ is fresh for $M$. 

The following results about lists will enable us to construct and destruct derivations of \free{x}{M} and \fresh{x}{M} as if they were inductively defined:
\begin{lemma}
\label{lemma:genList}
\begin{enumerate}[noitemsep]
\item \Verb|appIn : ∀ {x} {xs ys : List 𝒱} → x ∈ xs ++ ys ↔ x ∈ xs ⊎ x ∈ ys|
\item \Verb|delIn : ∀ {x y xs} → x ∈ xs - y ↔ x ≢ y × x ∈ xs|
\item \Verb|appNotIn : ∀ {x} {xs ys : List 𝒱} → x ∉ xs ++ ys ↔ x ∉ xs × x ∉ ys|
\item \Verb|delNotIn : ∀ {x y xs} → x ∉ xs - y ↔ (x ≡ y ⊎ x ∉ xs)|
\end{enumerate}
\end{lemma}
\noindent where \Verb|_⊎_| is disjoint union, \Verb|_×_| non-dependent product 
and \Verb|A ↔ B = (A → B) × (B → A)|. 

\subsection{The Choice Function}

To implement Stoughton's substitution operation, first we need to define a function that \textit{chooses} a fresh name for a list of given names.
This list will contain the free names in the image of the substitution which need to be avoided when renaming the bound variables. 

In previous version of the framework, names have been identified with the natural numbers, so the choice function had type:
\Verb|χ' : List ℕ → ℕ|.
There it was shown that such function actually returns a number fresh for the input list:
\begin{lemma}
\label{lemma:chiPrimeOriginal}
\Verb|xpfresh : ∀ xs → χ' xs ∉ xs|
\end{lemma}
\noindent Now in this development, the names in the list have an abstract type $\mathcal{V}$, hence we also need a more abstract choice function. To this end, we define a new function $X'$ such that, given a list $xs$, it encodes $xs$, then calls the previous function and finally decodes the result. Formally:

\begin{Verbatim}
Χ' : List 𝒱 → 𝒱
Χ' xs = from (χ' (map to xs))
\end{Verbatim}
It follows that this new function also selects a fresh name for the input list:

\begin{lemma}
\label{lemma:chiPrime}
\Verb|Xpfresh : ∀ xs → X' xs ∉ xs|
\end{lemma}

\begin{proof}
By \cref{lemma:chiPrimeOriginal} 
and by the fact that \Verb|decode| is the right inverse of \Verb|encode|.
\end{proof}

\subsection{Substitutions}

Substitutions are functions from variables to terms:
\begin{Verbatim}
Sub = 𝒱 → Λ
\end{Verbatim}
\noindent We shall use the letter $\sigma$ possibly with primes to range them.
We define \Verb|ι| as the \textit{identity} substitution, 
and an \textit{update} operation on substitutions \Verb|_,_:=_ : Sub → 𝒱 → Λ → Sub| such that for any given substitution $\sigma$, $\lambda$-term $N$ and names $x$ and $y$, $(\sigma,x:=N)\,y$ yields $N$ if $x$ equals to $y$ and $\sigma y$ otherwise.

To reason about substitutions it turns out convenient to confine their domain to finite subsets of them.
In previous work, the authors introduced the \emph{restriction} type, whose objects were just pairs of substitutions and $\lambda$-terms $(\sigma,M)$ and which would be used anywhere it was required to reason about the extension of $\sigma$ only up to the free names in $M$. 
In our case, however, we shall that see we will need a more flexible definition, therefore we define restrictions as pairs of substitutions and \emph{list of names}:
\begin{Verbatim}
Res = Sub × List 𝒱
\end{Verbatim}
For instance, we can use them to extend the choice function to substitutions:
\begin{Verbatim}
X : Res → 𝒱
X (σ , xs) = X' (concat (map (fv ∘ σ) xs))
\end{Verbatim}
\noindent where \Verb|concat| is the function that flattens a list of lists.

With our machinery for restrictions we can prove this extended choice function \Verb|X| also returns a fresh name.
Let the relations of free and fresh variables be extended to restrictions by:
\begin{Verbatim}
x $*$⇂ (σ , xs) = ∃ λ y → y ∈ xs × x ∈ fv (σ y) 
x #⇂ (σ , xs) = ∀ y → y ∈ xs → x # σ y 
\end{Verbatim}
Then we have that $X(\sigma,xs)$ is fresh for every $y \in xs$:
\begin{lemma}
\label{lemma:chi}
\Verb|Xfresh : ∀ σ xs → X (σ , xs) #⇂ (σ , xs)|
\end{lemma}

\begin{proof}
By using lists properties and \cref{lemma:chiPrime}.
\end{proof}

Without further ado, we define by recursion on the syntax the substitution operation, which given any substitution $\sigma$ and $\lambda$-term $M$, it replaces the free names in $M$ by their corresponding images in $\sigma$ while renaming the bound variables at the same time to prevent any possible name capture:
\begin{Verbatim}
_∙_ : Λ → Sub → Λ
c k ∙ _ = c k
v x ∙ σ = σ x
M · N ∙ σ = (M ∙ σ) · (N ∙ σ)
λ[ x ∶ A ] M ∙ σ = λ[ y ∶ A ∙ σ ](M ∙ σ , x := v y) where y = X (σ , fv M - x)  
Π[ x ∶ A ] B ∙ σ = Π[ y ∶ A ∙ σ ](B ∙ σ , x := v y) where y = X (σ , fv B - x)
\end{Verbatim}
In the equation for $\lambda$-abstractions we could have defined 
\Verb|y = X (σ , λ[ x ∶ A ] M)|, 
and analogously in the case of $\Pi$-types, and therefore saved us the trouble of revisiting the definition of restrictions.
Nevertheless, we find this to be rather unsatisfactory because in that case we would be excluding from consideration those names that occur free in the image of $A$, but not necessarily in that of $\fv{M}-x$, and which are safe candidates to rename the bound variable.
Hence we would not be choosing the \emph{first} name available but \emph{some} one.
We take the view that selecting the first name available is the best solution because it is the most straightforward way to formally specify such operation.
Besides, this way we are aligned with the formal treatment of the $\lambda$-calculus in the literature \cite{curry1958,hindley1986}.

We shall use the next abbreviation for unary substitution:
\Verb|M [ x := N ] = M ∙ ι , x := N|.

The first result we have concerning the substitution operation is that it is actually capture-avoiding:
\begin{lemma}
\label{lemma:nocapture}
\Verb|noCapture : ∀ {x M σ} → x ∈ fv (M ∙ σ) ↔ x $*$⇂ (σ , fv M)|
\end{lemma}
\noindent Read right-to-left: if $x$ is free in the image of $\sigma$ then it will remain so in the final result. 
The other direction is important as well, and it establishes that no free names are introduced by the operation other than those occurring in the image of the $\sigma$ (restricted to the free names in $M$).
 
Next we have some other nice properties on substitutions that we will use thoroughly in the next sections.
First, let equality on restrictions be defined by:
\begin{Verbatim}
(σ , xs) ≅⇂ (σ' , xs') = (xs ⊆ xs' × xs' ⊆ xs) × ∀ x → x ∈ xs → σ x ≡ σ' x
\end{Verbatim}
where list inclusion \Verb|_⊆_| is defined in the standard library by:
\Verb|xs ⊆ ys = ∀ {x} → x ∈ xs → x ∈ ys|.
As a special case of equality on restrictions we define: 
\Verb|σ ≅ σ' ⇂ xs = (σ , xs) ≅⇂ (σ' , xs)|.
Also, let composition of substitutions be defined by:
\begin{Verbatim}
(σ ⊙ σ') x = σ' x ∙ σ
\end{Verbatim}
Then we have the next results which are extended from \cite{copello2016,tasistro2015,urciuoli2023}:
\begin{lemma}
    \label{lemma:substitutionProperties}
    \begin{lemenum}
        \item 
        \label{lemma:subEqRes}
        \Verb|subEqRes : ∀ {M σ σ'} → σ ≅ σ' ⇂ fv M → M ∙ σ ≡ M ∙ σ'|
        
        \item
        \label{lemma:updFresh}
        \Verb|updFresh : ∀ {x M N σ} → x # M → M ∙ σ , x := N ≡ M ∙ σ|
        
        \item 
        \label{lemma:exp}
        \Verb|composRenUpd : ∀ {x z M N σ} → z ∉ fv M - x|\\
        \Verb| → M ∙ σ , x := N ≡ M [ x := v z ] ∙ σ , z := N|
        
        \item
        \label{lemma:subDistribUpd}
        \Verb|subDistribUpd : ∀ {M N σ x} → M ∙ σ , x := (N ∙ σ) ≡ M [ x := N ] ∙ σ|
        
        \item
        \label{lemma:subComp}
        \Verb|subComp : ∀ {M σ σ'} → M ∙ σ ∙ σ' ≡ M ∙ σ' ⊙ σ|        
    \end{lemenum}
\end{lemma}

\subsection{Alpha-conversion}
\label{sec:alpha}
We define $\alpha$-conversion inductively by:
\begin{Verbatim}
data _∼α_ : Λ → Λ → Set where
  ∼c : ∀ {k} → c k ∼α c k
  ∼v : ∀ {x} → v x ∼α v x
  ∼· : ∀ {M M' N N'} → M ∼α M' → N ∼α N' → M · N ∼α M' · N'
  ∼λ : ∀ {x x' y A A' M M'} → A ∼α A' → y ∉ fv M - x → y ∉ fv M' - x' 
     → M [ x := v y ] ≡ M' [ x' := v y ] → λ[ x ∶ A ] M ∼α λ[ x' ∶ A' ] M'
  ∼Π : ∀ {x x' y A A' B B'} → A ∼α A' → y ∉ fv B - x → y ∉ fv B' - x' 
     → B [ x := v y ] ≡ B' [ x' := v y ] → Π[ x ∶ A ] B ∼α Π[ x' ∶ A' ] B'
\end{Verbatim}

\noindent In the rule for $\lambda$-abstractions and $\Pi$-types we require that the bodies are syntactical equal once their respective bound names have been replaced by a common fresh name. In former work the premises were stated in terms of the very same relation being defined. 
Later in this section and by having a robust infrastructure we will show a posteriori
that both definitions for $\alpha$-conversion are equivalent.

By using this formulation we can show the following result which is central in the theory of $\alpha$-conversion with simpler methods than in previous work:
\begin{lemma}
\label{lemma:characterizatinoAlphaIota}
\Verb|iotaAlpha : ∀ {M M'} → M ∙ ι ≡ M' ∙ ι → M ∼α M'|
\end{lemma}

\newcommand{\subs}{\ensuremath{\,\bullet\,}}

\begin{proof}
By structural induction on $M$ and subordinate case analysis on $M'$.
We only show the sub-case for the $\lambda$-abstractions which did not follow by structural induction before. The respective sub-case for $\Pi$-types is analogous. There, given $\lambda[x:A]M\subs\iota \equiv \lambda[x':A']M'\subs\iota$, we must show that $\lambda[x:A]M \alphaconv \lambda[x':A']M'$. We proceed as follows. First, by definition of substitution we have that our hypothesis is definitional equal to $\lambda[z:A
\subs\iota]M[x:=z] \equiv \lambda[x':A'\subs\iota]M'[x':=z']$ for some sufficiently fresh names $z$ and $z'$. Next, by injectivity of the constructors we know that $z \equiv z'$, $A\subs\iota \equiv A'\subs\iota$ and $M[x:=z] \equiv M'[x':=z']$. Then, by the inductive hypothesis on $A$ we have $A \alphaconv A'$. 
And finally, by the rule of $\lambda$-abstractions we can derive the desired goal. 
\end{proof}

\noindent Notice that we did not have to use the induction hypothesis on $M[x:=y]$. 
Otherwise we would have to use some method other than structural induction, e.g., complete induction on the length of $M$ as in \cite{tasistro2015}, since $M[x:=y]$ is not a proper component of $\lambda[x:A]M$.

Next we have some properties about $\alpha$-conversion that are extended quite directly from previous work and which we are going to need later on .
Let $\alpha$-conversion be extended to restrictions by:
\begin{Verbatim}
σ ∼α σ' ⇂ xs = ∀ x → x ∈ xs → σ x ∼α σ' x
\end{Verbatim}
Then:
\begin{lemma}
    \label{lemma:alphaProperties}
    \begin{lemenum}
        \item 
        \label{lemma:equalization}
        \Verb|compatSubAlpha : ∀ {M M' σ} → M ∼α M' → M ∙ σ ≡ M' ∙ σ|
        
        \item
        \label{lemma:subAlpha}
        \Verb|subAlpha : ∀ {M σ σ'} → σ ∼α σ' ⇂ fv M → M ∙ σ ∼α M ∙ σ'|
        
        \item
        \label{lemma:alphaCongruence}
        \alphaconv\ is an equivalence.
        
        \item
        \label{lemma:factor}
        \Verb|composRenUnary : ∀ {x y σ M N} → y #⇂ (σ , fv M - x)|\\
        \Verb| → (M ∙ σ , x := v y) [ y := N ] ∼α M ∙ σ , x := N|
    \end{lemenum}
\end{lemma}

\newcommand{\refsublemma}[1]{Lemma~\ref{#1}}

\refsublemma{lemma:equalization} not only states that substitution is compatible with $\alpha$-conversion, but also that it \textit{equalizes} $\alpha$-convertible $\lambda$-terms due to the uniform renaming of the bound variables. Following \cite{stoughton88}, we will also say $\sigma$ puts $\lambda$-terms into $\sigma$-normal form with respect to $\alpha$-conversion or simply into $\sigma$-normal form.

Note that \refsublemma{lemma:factor} cannot be strengthened up to syntactical equality. On the left-hand side of \alphaconvg\ we have that the image of $\sigma$ for every name $z\not\equiv x$ is being $\iota$-normalized by the substitution $\iota , y := N$ (which has the same effect as $\iota$ since $y$ is fresh for every image in $\sigma$ by definition), while on the right-hand side it is not. This observation will have some repercussions on the proof of closure under substitution of the type system presented in \cref{sec:metatheory}.

\subsubsection{Soundness and Completeness of Alpha-conversion}

Let $\alphaconvg_s$ be defined as $\alphaconvg$ except that in the clause for $\lambda$-abstractions and $\Pi$-types we have the following premises for the bodies: $M[x:=y]\ \alphaconvg_s\ M'[x':=y]$ and $B[x:=y]\ \alphaconvg_s\ B'[x':=y]$ .
It follows that both definitions are extensionally equivalent.

As to the soundness of \alphaconvg, first we need the following result which can be easily adapted from previous work since it does not depend on any other result about $\alpha$-conversion nor substitutions (the proof is conducted by complete induction on the length of $M$ as already commented somewhere):

\begin{lemma}
\label{lemma:iotaAlphaSt}
\Verb|iotaAlphaSt : ∀ {M M'} → M ∙ ι ≡ M' ∙ ι → M ∼αₛ M'|
\end{lemma}

\begin{theorem}
\Verb|soundAlpha : ∀ {M N} → M ∼α N → M ∼αₛ N|
\end{theorem}

\begin{proof}
Let $M \sim_\alpha N$. By \refsublemma{lemma:equalization} we have $M \subs \iota \equiv N \subs \iota$, so by \cref{lemma:iotaAlphaSt} we have $M\, \alphaconvg_s\, N$.
\end{proof}

As to the other direction, first we need the following result:

\begin{lemma}
\label{lemma:alphaEq}
\Verb|alphaEq : ∀ {M M' x x' y} → M [ x := v y ] ∼α M' [ x' := v y ]|\\
\Verb| → M [ x := v y ] ≡ M' [ x' := v y ]|
\end{lemma} 

\theoremstyle{remark}
\newtheorem*{proof*}{Proof}

\begin{proof}
Note that by definition of equality on restrictions we have that $(\iota,x:=y)\cong \iota \odot (\iota,x:=y) \downharpoonright xs$ for any $x$, $y$ and $xs$.
Then we can reason in the following way: 
\begin{align*}
M[x:=y]& \equiv M \subs \iota \odot (\iota,x:=y) & & \text{by \refsublemma{lemma:subEqRes}}\\
&\equiv M[x:=y] \subs \iota & & \text{by \refsublemma{lemma:subComp}}\\
&\equiv M'[x':=y] \subs \iota & & \text{by \refsublemma{lemma:equalization}}\\
&\equiv M' \subs \iota \odot (\iota,x':=y)& & \text{by \refsublemma{lemma:subComp}}\\
&\equiv M'[x':=y]& & \text{by \refsublemma{lemma:subEqRes}}\qedhere
\end{align*}
\end{proof}

\begin{theorem}
\label{theo:adequacyAlpha}
\Verb|completeAlpha : ∀ {M N} → M ∼αₛ N → M ∼α N|
\end{theorem}
\begin{proof}
By structural induction on the derivation of $M\ \alphaconvg_s\ N$. The only interesting cases are that of $\lambda$-abstractions and $\Pi$-types. There, one has to use \cref{lemma:alphaEq} to obtain the desired premises.
\end{proof}

\subsection{Beta-conversion}

Given any binary relation on the $\lambda$-terms, 
$\_\mathcal{S}\_$, we define its contextual closure by:
\begin{Verbatim}
data _→C_ : Λ → Λ → Set where
  →cxt : ∀ {M N} → M 𝒮 N → M →C N    
  →λR  : ∀ {x M M' A} → M →C M' → λ[ x ∶ A ] M →C λ[ x ∶ A ] M'
  →ΠR  : ∀ {x B B' A} → B →C B' → Π[ x ∶ A ] B →C Π[ x ∶ A ] B'
  →λL  : ∀ {x M A A'} → A →C A' → λ[ x ∶ A ] M →C λ[ x ∶ A' ] M    
  →ΠL  : ∀ {x B A A'} → A →C A' → Π[ x ∶ A ] B →C Π[ x ∶ A' ] B
  →·L  : ∀ {M N P} → M →C N → M · P →C N · P
  →·R  : ∀ {M N P} → M →C N → P · M →C P · N 
\end{Verbatim}
$\beta$-contraction (\Verb|_▹β_|) is defined as the inductive relation with the single rule: 
$\lambda[x:A]M \cdot N\ {▹β}\ {M[x:=N]}$. 
Then one-step $\beta$-reduction (\Verb|_→β_|) is defined as its contextual closure: \Verb|_→β_ = _→C_ _▹β_|. Many-step $\beta$-reduction (\Verb|_→β$*$_|) and $\beta$-conversion (\Verb|_≃β_|) are respectively defined as the symmetric-and-transitive and equivalence closure of 
$\beta$-reduction augmented with $\alpha$-conversion: \Verb|_→β$*$_ = Star (_∼α_ ∪ _→β_)| and \Verb|_≃β_ = EqClosure (_∼α_ ∪ _→β_)|.

Our definition of one-step $\beta$-reduction is closed under substitution only up to $\alpha$-conversion as in Hindley and Seldin's treatment, thought there this result is left implicit. The proof is adapted from \cite{urciuoli2023}:
\begin{lemma}
\Verb|compatRedSub : ∀ {M N σ} → M →β N → ∃ λ P → M ∙ σ →β P × P ∼α N ∙ σ|
\end{lemma}
Many-step $\beta$-reduction and $\beta$-conversion, on the other hand, are ``fully'' closed under substitution and their proofs follow by structural induction on the derivations of the respective hypotheses, and which are also extended from the aforementioned cite:
\begin{lemma}
\begin{lemenum}
\item \Verb|compatRedsSub : ∀ {M N σ} → M →β$*$ N → M ∙ σ →β$*$ N ∙ σ|
\item 
\label{lemma:compatbetasub}
\Verb|compatConvSub : ∀ {M N σ} → M ≃β N → M ∙ σ ≃β N ∙ σ|
\end{lemenum}
\end{lemma}

\section{The Pure Type Systems}
\label{sec:calculus}

Let $\mathcal{A}$, the \emph{axioms}, be any binary relation on $\mathcal{C}$, and let $\mathcal{R}$, the \emph{rules}, be any 3-ary relation on $\mathcal{C}$ as well. Whenever some $s_1$ is related to some other $s_2$ under $\mathcal{A}$ we shall write $\mathcal{A}\,s_1\,s_2$ as is the case for infix relations in most implementations of type theory, e.g., Agda. We will do analogously for $\mathcal{R}$.

The PTS is inductively defined in \Cref{fig:system}, where $\Gamma , x : A$ is definitionally equal to $(x,A) \dblcolon \Gamma$.
We roughly follow the presentation in Section~4.4.10 of \cite{pollackPhD}. We have two kinds of judgments mutually defined and with the following meaning:
\begin{enumerate*}[label=(\roman*)]
\item \ok{\Gamma} means that $\Gamma$ is a valid or well-formed context, and;
\item $\Gamma \vdash M : A$ that $M$ has type $A$ under the context $\Gamma$.
\end{enumerate*}

\newcommand{\app}{\ensuremath{\cdot}}
\def\defaultHypSeparation{\hskip .1in}

\newcommand{\abstree}[1][\typerule{abs}]{
    \begin{prooftree}
        \AxiomC{$\Gamma \vdash A : s_1 \quad \forall z \rightarrow z\not\in\dom{\Gamma} \rightarrow \Gamma, z : A \vdash B[y:=z] : s_2$}
        \noLine
        \UnaryInfC{$\forall z \rightarrow z\not\in\dom{\Gamma} \rightarrow \Gamma, z : A \vdash M[x:=z] : B[y:=z]$}   
        \LeftLabel{\small#1}
        \RightLabel{($\mathcal{R}\,s_1\,s_2\,s_3$)}
        \UnaryInfC{$\Gamma \vdash \lambda[ x : A ] M : \Pi[ y : A ] B$}
    \end{prooftree}
}

\newcommand{\apptree}[1][\typerule{app}]{
    \begin{prooftree}
        \AxiomC{$\Gamma \vdash M : \Pi[ x : A ] B$}
        \AxiomC{$\Gamma \vdash N : A$}
        \AxiomC{$\Gamma \vdash B[x:=N]:s$}    
        \LeftLabel{\small#1}
        \TrinaryInfC{$\Gamma \vdash M \app N : B[x:=N]$}
    \end{prooftree}
}

\newcommand{\vartree}[1][\typerule{var}]{
    \begin{prooftree}
        \AxiomC{\ok{\Gamma}}
        \LeftLabel{\small#1}
        \RightLabel{($(x,A) \in \Gamma$)}
        \UnaryInfC{$\Gamma \vdash x : A$}
        \end{prooftree}
}

\begin{figure}[t]
\small
\centering
\AxiomC{}
\LeftLabel{\small\typerule{nil}} 
\UnaryInfC{$\ok{[]}$}
\bottomAlignProof
\DisplayProof
\quad
\AxiomC{\ok{\Gamma}}
\AxiomC{$\Gamma \vdash A : s$}
\LeftLabel{\small\typerule{cons}} 
\RightLabel{($x \not\in \dom{\Gamma}$)}
\BinaryInfC{\ok{\Gamma,x:A}}
\bottomAlignProof
\DisplayProof

\bigskip

\AxiomC{$\ok{\Gamma}$}
\LeftLabel{\small\typerule{sort}}
\RightLabel{($\mathcal{A}\,s_1\,s_2$)}
\UnaryInfC{$\Gamma \vdash s_1 : s_2$}
\bottomAlignProof
\DisplayProof

\bigskip

\AxiomC{$\Gamma \vdash A : s_1$}
\AxiomC{$\forall y \rightarrow y \not\in \dom{\Gamma} \rightarrow \Gamma,y:A \vdash B[x:=y] : s_2$}
\LeftLabel{\small\typerule{prod}}
\RightLabel{($\mathcal{R}\,s_1\,s_2\,s_3$)}
\BinaryInfC{$\Gamma \vdash \Pi[x:A]B : s_3$}
\bottomAlignProof
\DisplayProof

\vartree

\abstree

\apptree

\AxiomC{$\Gamma \vdash M : A$}
\AxiomC{$A \betaconv B$}
\AxiomC{$\Gamma \vdash B : s$}
\LeftLabel{\small\typerule{conv}}
\TrinaryInfC{$\Gamma \vdash M : B$}
\bottomAlignProof
\DisplayProof
\caption{Pure Type Systems}
\label{fig:system}
\end{figure}

To be more precise, the type system is defined by using \emph{generalized} induction \cite[p.~382]{mckinna99}. 
The rules \typerule{abs} and \typerule{prod} have infinitely many branching trees due to the scope of the quantification in the premises concerning the body of the abstractions therein.
Having a derivation for every fresh name will give us a stronger induction hypothesis in the cases at issue for the proof of the thinning lemma, 
which otherwise we would have to prove either by using some equivariance result (a renaming lemma) or some other more involved method.
A more up-to-date version of this technique is known as \textit{cofinite quantification}.
We note that Agda does not give any special meaning to generalized induction and it is treated just as any other inductive clause. We will also refer to these kind of rules as \textit{infinitary}.

The occurrence of the premise $\Gamma \vdash A : s_1$ in the rule \typerule{prod} is standard, e.g., see \cite{barendregt92}. On the other hand, having $\Gamma \vdash B[x:=N] : s$ in the \typerule{app} rule is new to the best of our knowledge, and as we shall see in due course, it is convenient for the proof of the substitution lemma when using Stoughton's definition.

\subsection{Basic Properties}

To begin with, we have that variables mentioned anywhere must be declared:

\begin{lemma}
\label{lemma:declaration}
\begin{lemenum}
\item \Verb|fvCxt : ∀ {Γ x A w} → Γ ok → (x , A) ∈ Γ → w $*$ A → w ∈ dom Γ|
\item \Verb|fvAsg : ∀ {Γ M A w} → Γ ⊢ M ∶ A → x $*$ M · A → w ∈ dom Γ|
\end{lemenum}
\end{lemma}

\begin{proof}
By simultaneous induction on the the structure of the typing derivation and subordinate case analysis on the occurrence of $w$. The only interesting case is that of $\lambda$-abstractions; the others follow either analogously or directly by the IH.
There, we have the hypotheses:
\abstree[\normalsize(a)]
and (b) $w \in \fv{(\lambda[x : A]M \appendg \Pi[x : A]B)}$, and we must show $w \in \dom{\Gamma}$. By \cref{lemma:genList} there are three different cases as to the generation of (b):
\begin{itemize}
	\item Case $w*A$. Immediate by the IH (ii).
	\item Case $w*M$ and $x \not\equiv w$. 
			Let $z=X'(w \dblcolon \dom{\Gamma})$. 
			By \cref{lemma:chiPrime} we have that $z\not\in w\dblcolon\dom{\Gamma}$, thus $z\not\in\dom{\Gamma}$ and $z\not\equiv w$. 
			Next, since $w * w[x:=z]$ then we have $w * M[x:=z]$ by \cref{lemma:nocapture}. 
			Finally, by the IH~(ii) using the bottom-most premise we have $w\in\dom{z\dblcolon\Gamma}$, hence $w\in\dom{\Gamma}$.
			\item Case $w*B$ and $y \not\equiv w$. Analogous to the previous case.\qedhere
\end{itemize}
\end{proof}

The next result follows directly as the contrapositive of the previous lemma:

\begin{corollary}
\label{lemma:freshness}
\begin{lemenum}
\item \Verb|freshCxt : ∀ {Γ y A w} → Γ ok → w ∉ dom Γ → (y , A) ∈ Γ → w # A|
\item \Verb|freshAsg : ∀ {Γ M A w} → w ∉ dom Γ → Γ ⊢ M ∶ A → w # M · A|
\end{lemenum}
\end{corollary}

Lastly, the following results about the validity of contexts and generation of $\Pi$-types are routine:

\begin{lemma}
    \begin{lemenum}
        \item
        \label{lemma:contextValidityAsg}
        \Verb|validCxt : ∀ {Γ M A} → Γ ⊢ M ∶ A → Γ ok|
        \item 
        \label{lemma:genProd}
        \Verb|genProd :  ∀ {Γ x A B C} → Γ ⊢ Π[ x ∶ A ] B ∶ C → ∃₃ λ s₁ s₂ s₃ → ℛ s₁ s₂ s₃|\\
        \Verb| × Γ ⊢ A ∶ c s₁ × (∀ y → y ∉ dom Γ → Γ ‚ y ∶ A ⊢ B [ x := v y ] ∶ c s₂) × C ≃β c s₃|
    \end{lemenum}
\end{lemma}

\section{Some Fundamental Metatheory}
\label{sec:metatheory}

The type system presented in the previous section enjoys some nice metatheoretic properties such as thinning (\cref{lemma:weakening}), syntactic validity (\cref{lemma:typeValidity}), closure under $\alpha$-conversion (\cref{lemma:closureAlpha}) and substitution (\cref{lemma:substitution}). 

The thinning lemma follows quite easily thanks to the use of generalized induction. 

As to the other results, the situation is a bit more complicated.
Consider a presentation of the type system in which the application rule does not mention the third premise, and let use proceed to prove informally, to begin with, syntactic validity and by using structural induction.
In the complex case of applications we are given:
\begin{equation}
\label{eq:app}
\AxiomC{$\Gamma \vdash M : \Pi[x:A]B$}
\AxiomC{$\Gamma \vdash N : A$}
\BinaryInfC{$\Gamma \vdash M \app N : B[x:=N]$}
\DisplayProof
\end{equation}
and we have to show that $B[x:=N]$ is a valid type, i.e., either $\Gamma \vdash B[x:=N]:s$ or $B[x:=N]\equiv s$ for some sort $s$. 
First, by the induction hypothesis we can derive that $\Gamma\vdash\Pi[x:A]B:s$.
Then, by the generation lemma we get $\Gamma,y:A \vdash B[x:=y]:s$ for every fresh name $y$. 
Now, to derive the goal from this last result we  have to apply the substitution $\iota,y:=N$ and obtain a valid type, all of which requires having some substitution lemma at hand. So far is routine.

Next let us consider the proof of closure under substitution of the type system.
Again, let us take a look at the case of applications.
Given (\ref{eq:app}), we have to show $\Delta \vdash (M \app N) \subs \sigma:B[x:=N]\subs\sigma$ for some context $\Delta$. 
By the induction hypothesis on each premise 
followed by
the application rule one can derive ${\Delta \vdash (M \app N)\subs\sigma:(B\subs\sigma,x:=x')[x':=N\subs\sigma]}$. 
Now, the types ${(B\subs\sigma,x:=x')[x':=N\subs\sigma]}$ and $B[x:=N]\subs\sigma$ are not identical but $\alpha$-convertible, 
as already commented somewhere,
so it is necessary to have closure under $\alpha$-conversion at hand.
In \cite{mckinna93,mckinna99,pollackPhD} the authors were able to prove directly the lemma because there substitution preserves the identity of the $\lambda$-terms. 
However, as we shall address later on, they had to resort to more involved methods than us for an important lemma in the theory of conversion.

An alternative to using closure under $\alpha$-conversion above is to 
use the conversion rule and rewrite the types at question, however, this is not easy.
In that case, one should first 
build a derivation of the (syntactic) validity of 
${(B\subs\sigma,x:=x')[x':=N\subs\sigma]}$.
The direct procedure 
to this end 
is to use the validity lemma on the left-hand side premise in (\ref{eq:app}), then the generation lemma to derive that $B[x:=y]$ is well-formed for some fresh name $y$ and finally the induction hypothesis twice, once with $\sigma,y:=z$ and other time with $\iota,z:=N\subs\sigma$ for some sufficiently fresh name $z$ (which might perfectly be the same name as $y$), 
all of which yields a result identical to the one desired because of \refsublemma{lemma:exp}.\footnote{Note that we do not have a result for composing well-typed substitutions, a result which seems to require the very same lemma we are trying to prove.}
Now, the main problem with this argument is that it is not always the case that the derivation of the validity of $B[x:=y]$ constructed during the proof of the validity lemma is of a smaller size in any way than that of (\ref{eq:app}), hence the well-foundedness of the whole argument is questioned.
A more elaborated method is required.

So, back to our point, the proof of closure under $\alpha$-conversion also relies on the previous lemmas. In the case of applications we have $M\app N \alphaconv M'\app N'$, in addition to (\ref{eq:app}), and we must find a derivation of $\Delta \vdash M'\app N' : B[x:=N]$ for some context $\Delta$. By the induction hypotheses followed by the application rule we can derive $\Delta \vdash M' \app N' : B[x:=N']$. Now, to rewrite the type $B[x:=N']$ into $B[x:=N]$, first we have to show that the latter is a valid type. This only seems possible if we use syntactic validity and then substitution, and so closing the dependency circle.

Now, instead of attempting to prove all three lemmas simultaneously, we will make a slight simplification to the problem.
By adding the premise $\Gamma \vdash B[x:=N]:s$ to the application rule we have managed to cut two of the dependencies: $\alpha$-conversion on syntactic validity, and the latter on substitution. As a result, we shall see that syntactic validity and $\alpha$-conversion can be proven separately and first, and substitution after.
Furthermore, we shall prove that said premise is derivable in the sense that we can consider a presentation of the type system that does not mention it yet derives the same judgments.

\subsection{Thinning}

To begin with, we have the thinning lemma which follows in a fairly direct way (we refer the reader to \cite[p.~65]{pollackPhD} for an account on the proof):

\begin{lemma}
\label{lemma:weakening}
\Verb|thinning : ∀ {Γ Δ M A} → Γ ⊆ Δ → Δ ok → Γ ⊢ M ∶ A → Δ ⊢ M ∶ A|
\end{lemma}

\newcommand{\lkup}[2]{\ensuremath{#1\langle#2\rangle}}
\newcommand{\const}[1]{\ensuremath{\mathsf{c\,#1}}}

\subsection{Syntactic Validity}

Every type assignable to some $\lambda$-term is itself a valid type. For many presentations, e.g., \cite{barendregt92,harper93,mckinna93,mckinna99,pollackPhD}, this result has to wait until closure under substitution has been established. In our case, however, because we have added the aforementioned premise we can show it in advance. 

The proof is really straightforward and we do not need to use induction, just case analysis. We only need the next lemma for the case of variables and which follows by structural induction on the derivation of \ok{\Gamma} and by using the thinning lemma: 

\begin{lemma}
\label{lemma:validDecl}
\Verb|validDecl : ∀ {Γ x A} → Γ ok → (x , A) ∈ Γ → ∃ λ s → Γ ⊢ A ∶ c s|
\end{lemma}

\begin{lemma}
\label{lemma:typeValidity}
\Verb|syntacticValidity : ∀ {Γ M A} → Γ ⊢ M ∶ A → ∃ λ s → A ≡ c s ⊎ Γ ⊢ A ∶ c s|
\end{lemma}

\subsection{Closure Under Alpha-conversion}

Next we have closure under $\alpha$-conversion. We are going to split the lemma in half, one for the conversion of subjects and the other for the predicates.

Let $\alpha$-conversion be extended to contexts pointwise by:
\begin{Verbatim}
_≈α_ = Pointwise (λ (x , A) (y , B) → x ≡ y × A ∼α B)
\end{Verbatim}
\noindent Then we have the first part:
\begin{lemma}
\label{lemma:closureAlpha}
\begin{lemenum}
\item \Verb|closAlphaCxt : ∀ {Γ Δ} → Γ ≈α Δ → Γ ok → Δ ok|
\item \Verb|closAlphaAsg : ∀ {Γ Δ M N A} → Γ ≈α Δ → M ∼α N → Γ ⊢ M ∶ A → Δ ⊢ N ∶ A|
\end{lemenum}
\end{lemma}

\newcommand{\alphaconvcxt}{\ensuremath{\ {\approx}\alpha}\ }
\newcommand{\alphaconvcxtg}{\ensuremath{{\approx}\alpha}}

\begin{proof}
By simultaneous induction on the structure of the typing derivation.
We only show some cases:

\begin{itemize}
    \item Case \typerule{abs}.
    We have:
    \abstree[]
    and $\lambda[x:A]M \alphaconv \lambda[x':A']M'$ which follows from  $A\alphaconv A'$ and 
    ${M[x:=w]\equiv M'[x':=w]}$ 
    for some $w$ not in $\fv{M}-x$ nor in $\fv{M'}-x'$. We have to show: $\Delta \vdash \lambda[x':A']M' : \Pi[y:A]B$. 
    To use the $\lambda$-abstraction rule to derive $\Delta \vdash \lambda[x':A']M' : \Pi[y:A']B$, and so then use the conversion rule to rewrite $\Pi[y:A']B$ into $\Pi[y:A]B$, first we need to show:
    \begin{enumerate}[label=(\alph*)]
        \item $\Delta \vdash A : s_1$;
        \item $\Delta,z:A' \vdash B[y:=z] : s_2$ for all $z\not\in\dom{\Delta}$, and;
        \item $\Delta,z:A' \vdash M[x:=z] : B[y:=z]$ for all $z\not\in\dom{\Delta}$.
    \end{enumerate}
    (a) is immediate by the IH (ii).
    As to (b) and (c), we show only the latter since the other is analogous. Let $z$ be some name not in $\dom{\Delta}$. Then we also have that $z\not\in\dom{\Gamma}$ by definition of \alphaconvcxtg.
    Next, by congruence on one of the hypotheses we have
    ${M[x:=w][w:=z]\equiv M'[x':=w][w:=z]}$,
    and by \refsublemma{lemma:exp}
    we can obtain $M[x:=z]\equiv M'[x':=z]$.
    Then, since $\alpha$-conversion is reflexive, we have
    $M[x:=z]\alphaconv M'[x':=z]$,
    so we can use the IH (ii) and get ${\Delta \vdash \lambda[x':A']M' : \Pi[y:A']B}$. Now, it is easy to show that $\Pi[y:A]B\alphaconv\Pi[y:A']B$, hence by the IH (ii) we have $\Delta\vdash\Pi[y:A']B$, and so we can use the conversion rule as explained earlier.
    \item Case \typerule{app}. 
    There we have:
    \apptree[]
    $M\alphaconv M'$ and $N\alphaconv N'$, and we must show $\Delta \vdash M' \app N' : B[x:=N]$.
    First, note that $\iota,x:=N$ and $\iota,x:=N'$ assigns $\alpha$-convertibles terms to every variable, thus we have $\iota,x:=N \alphaconv \iota,x:=N' \downharpoonright \mathsf{fv}\,B$. 
    Then, by \refsublemma{lemma:subAlpha} we obtain $B[x:=N'] \alphaconv B[x:=N]$.
    Now, by the IH~(ii) on each premise followed by the application rule we obtain ${\Delta \vdash M' \app N' : B[x:=N']}$. Finally, by the IH~(ii) again we can derive $\Delta \vdash B[x:=N] : s$ and so we can rewrite the types to obtain the desired goal. \qedhere
\end{itemize}
\end{proof}

As to the second part of the lemma 
we have that predicates are closed under $\alpha$-conversion as well. Its proof follows directly by using closure of the subjects and syntactic validity:

\begin{corollary}
\Verb|closAlphaPr : ∀ {Γ Δ M A B} → Γ ≈α Δ → A ∼α B → Γ ⊢ M : A → Δ ⊢ M : B|
\end{corollary}

\subsection{Closure Under Substitution}

Having established thinning, syntactic validity and closure under $\alpha$-conversion, we are almost ready to prove the substitution lemma. But first, some quick preparatory definitions. Let us define well-typed substitutions from variables in $\Gamma$ to terms of appropriate type under $\Delta$ by:
\begin{Verbatim}
σ ∶ Γ ⇀ Δ = ∀ {x A} → (x , A) ∈ Γ → Δ ⊢ σ x ∶ A ∙ σ
\end{Verbatim}
The next result will provide us the required hypothesis to invoke the induction hypothesis in the case of $\lambda$-abstractions and $\Pi$-types:

\begin{lemma}
\label{lemma:renaming}
\Verb|subRen : ∀ {Γ Δ x y s A σ} → x ∉ dom Γ → y ∉ dom Δ → Γ ⊢ A : c s|\\
\Verb| → Δ ⊢ A ∙ σ : c s → σ ∶ Γ ⇀ Δ → (σ , x := v y) ∶ (Γ ‚ x ∶ A) ⇀ (Δ ‚ y ∶ A ∙ σ)|
\end{lemma}

\begin{proof}
By definition of well-typed substitutions, we must show that for any declaration $(z,B)\in\Gamma,x:A$, it follows that $\Delta,y:A\subs\sigma \vdash (\sigma,x:=y)z : B\subs (\sigma,x:=y)$. 
First, notice that by \refsublemma{lemma:contextValidityAsg} we have \ok{\Delta}. Also, since $y\not\in\dom{\Delta}$, then by the $\vdash$cons rule we have \ok{\Delta,y:A\subs\sigma} as well.
Next we analyze whether $z$ is equal to $x$ or not:
\begin{itemize}
    \item Case $z \equiv x$. 
    Then $A \equiv B$ must also be the case, so the goal becomes ${\Delta,y:A\subs\sigma \vdash y : A\subs \sigma,x:=y}$. 
    By $\vdash$var we have $\Delta,y:A\subs\sigma \vdash y : A\subs\sigma$. Then, by the \cref{lemma:freshness} we have \fresh{x}{A}, hence by \refsublemma{lemma:updFresh} we know that $A\subs\sigma,x:=y\equiv A\subs\sigma$, and so we can rewrite the type in the previous derivation into $A\subs \sigma,x:=y$ to obtain our goal.
    
    \item Case $z \not\equiv x$. Then $(z,B)\in\Gamma$ and we must show $\Delta,y:A\subs\sigma \vdash \sigma z : B\subs\sigma,x:=y$. First, by hypothesis we have $\Delta \vdash \sigma z : A\subs\sigma$, so by thinning we also have $\Delta,y: A\subs\sigma \vdash \sigma z : B\subs\sigma,x:=y$. Then we can proceed analogously to the previous case and rewrite the type $B\subs\sigma,x:=y$ into $B\subs\sigma$.\qedhere
\end{itemize}
\end{proof}

\noindent Then we have that typing is closed under well-typed substitutions:

\begin{lemma}
\label{lemma:substitution}
\Verb|closureSub : ∀ {Γ Δ M A σ} → σ ∶ Γ ⇀ Δ → Δ ok → Γ ⊢ M ∶ A|\\
\Verb| → Δ ⊢ M ∙ σ ∶ A ∙ σ|
\end{lemma}

\begin{proof}
By simultaneous induction on the typing derivation. We only show some cases:
\begin{itemize}
    \item Case \typerule{abs}. We have:
    \abstree[]
    and we must derive $\Delta \vdash \lambda[x':A\subs\sigma](M\subs\sigma,x:=x'):\Pi[y':A\subs\sigma](B\subs\sigma,y:=y')$, where $x'$ and $y'$ are definitional equal to $X(\sigma,\fv{M}-x)$ and $X(\sigma,\fv{B}-y)$ respectively. To use the $\lambda$-abstraction rule to derive our goal first we need to show:
    \begin{enumerate}[label=(\alph*)]
    \item $\Delta\vdash A\subs\sigma : s_1$
    \item $\Delta , z\,{:}\,A\subs\sigma \vdash (B\subs\sigma,y \coloneq y')[y' \coloneq z] : s_2$ for all $z\not\in\dom{\Delta}$, and;
    \item $\Delta , z\,{:}\,A\subs\sigma \vdash (M\subs\sigma,x\,{:=}\,x')[x' \coloneq z] : (B\subs\sigma,y \coloneq y')[y' \coloneq z]$ for all $z\not\in\dom{\Delta}$.
    \end{enumerate}
    (a) is immediate from the IH. 
    As to (b) and (c), again, we only show the latter. Let $z$ be any name not in $\Delta$ and let $w=X'(\dom{\Gamma})$. 
    By \cref{lemma:chiPrime} we have $w\not\in\dom{\Gamma}$, thus by \cref{lemma:renaming} ${(\sigma,w:=z) : (\Gamma,w:A) \rightharpoonup (\Delta,z:A\subs\sigma)}$. Next, by the IH with the previous result we have:
    \begin{equation}
    \Delta,z: A\subs\sigma \vdash M[x:=w]\subs(\sigma,w:=z) : B[y:=w]\subs(\sigma,w:=z)
    \end{equation}
    Now, to derive our goal it suffices to show that the subjects in the goal (b) and in (2) are $\alpha$-convertible, and similarly with the predicates therein. As to the subjects, first,
    by \refsublemma{lemma:freshness} we note that $w\not\in\fv{M}-x$, thus by 
    \refsublemma{lemma:updFresh}
    we have $M[x:=w]\subs(\sigma,w:=z') \equiv M\subs(\sigma,x:=z')$.
    And second, by \cref{lemma:chi} we have ${x' \freshrg (\sigma,\fv{M}-x)}$, hence we can use \refsublemma{lemma:factor} and derive that ${M\subs(\sigma,x:=z) \alphaconv (M\subs\sigma,x:=x')[x':=z]}$. We can reason analogously to show that the types are also $\alpha$-convertible.
    Finally, by closure under $\alpha$-conversion we can obtain (b).
    
    \item Case \typerule{app}. We have:
    \apptree[]
    and we must show $\Delta \vdash M\app N\subs\sigma : B[x:=N]\subs\sigma$. By the IH on each premise followed by the application rule we have
    ${\Delta \vdash (M\app N)\subs\sigma : (B\subs\sigma,x:=x')[x':=N\subs\sigma]}$
    where $x'=X(\sigma,\fv{B}-x)$. Next, we can use \refsublemma{lemma:factor} and derive
    ${(B\subs\sigma,x:=x')[x':=N\subs\sigma] \alphaconv B\subs\sigma,x:=(N\subs\sigma)}$.
    Finally, by \refsublemma{lemma:subDistribUpd} we have
    $B\subs\sigma,x:=(N\subs\sigma) \equiv B[x:=N]
    \subs\sigma$ so by closure under $\alpha$-conversion we obtain our goal.

    \item Case \typerule{conv}. By the IH and \refsublemma{lemma:compatbetasub}.\qedhere
\end{itemize}
\end{proof}

As a particular case of the substitution lemma we have the following cut result.
%,
%which will turn out to be helpful in the following section.
First, let us extend the substitution operation pointwise to contexts by: \Verb|_∙∙_ : Cxt → Sub → Cxt|.
Then we have the next result for unary substitutions (whose proof is analogous to the proof of \cref{lemma:renaming}): 

\begin{lemma}
\label{lemma:unary}
\Verb|subUnary : ∀ {x Γ N A} → x ∉ dom Γ → Γ ⊢ A → Γ ∙∙ ι ⊢ N ∶ A ∙ ι|\\
\Verb| → (ι , x := N) ∶ (Γ ‚ x ∶ A) ⇀ Γ ∙∙ ι|
\end{lemma}

Then we have the cut lemma:

\begin{lemma}
\Verb|cut : ∀ {Γ M N A B x} → Γ ‚ x ∶ A ⊢ M ∶ B → Γ ⊢ N ∶ A|\\
\Verb| → Γ ⊢ M [ x := N ] ∶ B [ x := N ]|
\end{lemma}

\begin{proof}
By \cref{lemma:typeValidity}, closure under $\alpha$-conversion, \cref{lemma:unary}, context validity and \cref{lemma:substitution}.
\end{proof}

\subsection{A More Standard Presentation}

\newcommand{\abstreestd}{
    \begin{prooftree}
        \AxiomC{$\Gamma \stdpts A : s_1$}
        \AxiomC{$\Gamma, z : A \stdpts B[y:=z] : s_2$}
        \AxiomC{$\Gamma, z : A \stdpts M[x:=z] : B[y:=z]$}           
		\LeftLabel{\typerule{abs}}
		\RightLabel{$\begin{cases}\mathcal{R}\,s_1\,s_2\,s_3\\z\not\in\fv{M}-x\\z\not\in\fv{B}-y\end{cases}$}
        \TrinaryInfC{$\Gamma \stdpts \lambda[ x : A ] M : \Pi[ y : A ] B$}
    \end{prooftree}
}

\newcommand{\abstreestdNoLabels}{
    \begin{prooftree}
        \AxiomC{$\Gamma \stdpts A : s_1$}
        \AxiomC{$\Gamma, z : A \stdpts B[y:=z] : s_2$}
        \AxiomC{$\Gamma, z : A \stdpts M[x:=z] : B[y:=z]$}           
        \TrinaryInfC{$\Gamma \stdpts \lambda[ x : A ] M : \Pi[ y : A ] B$}
    \end{prooftree}
}

\newcommand{\apptreestd}[1][\typerule{app}]{
    \begin{prooftree}
        \AxiomC{$\Gamma \stdpts M : \Pi[ x : A ] B$}
        \AxiomC{$\Gamma \stdpts N : A$}
        \LeftLabel{\small#1}
        \BinaryInfC{$\Gamma \stdpts M \app N : B[x:=N]$}
    \end{prooftree}
}

\begin{figure}[t]
    \centering
    \small
    
    \AxiomC{}
    \LeftLabel{\small\typerule{nil}} 
    \UnaryInfC{$\stdok{[]}$}
    \bottomAlignProof
    \DisplayProof
    \quad
    \AxiomC{\stdok{\Gamma}}
    \AxiomC{$\Gamma \vdash A : s$}
    \LeftLabel{\footnotesize\typerule{cons}} 
    \RightLabel{($x \not\in \dom{\Gamma}$)}
    \BinaryInfC{\stdok{\Gamma,x:A}}
    \bottomAlignProof
    \DisplayProof
    
    \bigskip
    
    \AxiomC{$\stdok{\Gamma}$}
    \LeftLabel{\small\typerule{sort}}
    \RightLabel{($\mathcal{A}\,s_1\,s_2$)}
    \UnaryInfC{$\Gamma \vdash s_1 : s_2$}
    \bottomAlignProof
    \DisplayProof
    
    \bigskip
    
    \AxiomC{$\Gamma \stdpts A : s_1$}
    \AxiomC{$\Gamma,y:A \stdpts B[x:=y] : s_2$}
    \LeftLabel{\small\typerule{prod}}
    \RightLabel{\small$\begin{cases}\mathcal{R}\,s_1\,s_2\,s_3\\y \not\in\fv{B}-x\end{cases}$}
    \BinaryInfC{$\Gamma \stdpts \Pi[x:A]B : s_3$}
    \bottomAlignProof
    \DisplayProof
    
    \bigskip
    
    \AxiomC{\stdok{\Gamma}}
    \LeftLabel{\small\typerule{var}}
    \RightLabel{($(x,A) \in \Gamma$)}
    \UnaryInfC{$\Gamma \vdash x : A$}
    \DisplayProof
    
    \abstreestd
    
    \apptreestd
        
    \AxiomC{$\Gamma \vdash M : A$}
    \AxiomC{$A \betaconv B$}
    \AxiomC{$\Gamma \vdash B : s$}
    \LeftLabel{\small\typerule{conv}}
    \TrinaryInfC{$\Gamma \vdash M : B$}
    \DisplayProof
    \caption{Standard (Finitary) PTS}
    \label{fig:standardSystem}
\end{figure}

The type system defined in \cref{fig:system} is equivalent to one which uses finitary rules, i.e., without using generalized induction or cofinite quantification in the \typerule{prod} and \typerule{abs} rules, and furthermore, which does not mention the third premise in the application rule. 
\cref{fig:standardSystem} shows such a presentation.
For a brief discussion on why the freshness conditions in the finitary version are stated relative to the $\lambda$-terms therein while in the infinitary one are relative to the context we refer the reader to \cite[p.~61]{pollackPhD}.

First, we have that the infinitary version of the system is \textit{sound}:

\begin{theorem}
\begin{lemenum}
\item \Verb|ptsSound : ∀ {Γ} → Γ ok → Γ okₛ|
\item \Verb|ptsSound : ∀ {Γ M A} → Γ ⊢ M ∶ A → Γ ⊢ₛ M ∶ A|
\end{lemenum}
\end{theorem}

\begin{proof}
By straightforward induction on the structure of the judgments. In the case of products and abstractions, to derive the corresponding premises one has to pick some sufficiently fresh name, e.g., $X'(\dom{\Gamma})$, and use \cref{lemma:freshness} to derive the freshness side-conditions.
\end{proof}

The type system is also \textit{complete}. To prove it, we need the following renaming lemma (whose proof follows similarly to that of \cref{lemma:renaming}):

\begin{lemma}
\label{lemma:unaryRen}
\Verb|unaryRen : ∀ {Γ x y A M B} → y ∉ dom Γ → Γ ‚ x ∶ A ⊢ M ∶ B|\\
\Verb| → Γ ‚ y ∶ A ⊢ M [ x := v y ] ∶ B [ x := v y ]|
\end{lemma}

\begin{theorem}
\begin{lemenum}
\item \Verb|ptsComplete : ∀ {Γ} → Γ okₛ → Γ ok|
\item \Verb|ptsComplete : ∀ {Γ M A} → Γ ⊢ₛ M ∶ A → Γ ⊢ M ∶ A|
\end{lemenum}
\end{theorem}

\begin{proof}
By induction on the structure of the judgments.
\begin{itemize}
    \item Case \typerule{abs}.
    We have:
    \begin{center}
        \AxiomC{$\Gamma \stdpts A : s_1$}
        \AxiomC{$\Gamma, z : A \stdpts B[y:=z] : s_2$}
        \AxiomC{$\Gamma, z : A \stdpts M[x:=z] : B[y:=z]$}           
        \TrinaryInfC{$\Gamma \stdpts \lambda[ x : A ] M : \Pi[ y : A ] B$}
        \DisplayProof
    \end{center}
    with $z\not\in\fv{M}-x$ and $z\not\in\fv{B}-y$, 
    and we must derive $\Gamma \vdash \lambda[ x : A ] M : \Pi[ y : A ] B$.
    To use the abstraction rule we have to translate each premise. The left-most one is immediate from the IH~(ii). The other two are analogous so we will show only the right-most one. First, by the IH~(ii) we have $\Gamma , z:A \vdash M[x:=z] : B[y:=z]$. Let $z'$ be some name not in $\dom{\Gamma}$. By \cref{lemma:unaryRen} we have $\Gamma , z':A \vdash M[x:=z][z:=z'] : B[y:=z][z:=z']$, so by \refsublemma{lemma:exp} we have the desired result.
    
    \item Case \typerule{app}. 
    We have:
    \begin{center}
        \AxiomC{$\Gamma \stdpts M : \Pi[ x : A ] B$}
        \AxiomC{$\Gamma \stdpts N : A$}
        \BinaryInfC{$\Gamma \stdpts M \app N : B[x:=N]$}
        \DisplayProof
    \end{center}
    and we must derive $\Gamma \vdash M \app N : B[x:=N]$. By the IH~(ii) we have $\Gamma \vdash M : \Pi[x:A]B$ and $\Gamma \vdash N : A$. 
    Let $z=X'(\dom{\Gamma})$.
    By syntactic validity we obtain $\Gamma \vdash \Pi[x:A]B : s$ for some $s$, so by the generation lemma we derive $\Gamma,z:A \vdash B[x:=z] : s$. Now, by cut we have $\Gamma \vdash B[x:=z][z:=N] : s$, and so by \refsublemma{lemma:exp} we can derive the missing premise, $\Gamma \vdash B[x:=N] : s$.
\end{itemize}
\end{proof}

\section{Related Work}
\label{sec:related}

McKinna and Pollack \cite{mckinna93,mckinna99,pollackPhD} put forward in the LEGO proof assistant \cite{lego} the mechanization of a great body of knowledge about the metatheory of a generalization of the PTS, the Cumulative Type System (CTS): CR and standardization for $\beta$-conversion, subject reduction,
decidability of the type system (assuming normlization), etc. 
The syntax uses two sort of names, one for the free variables and other for the bound ones. Consequently, two definitions of the substitution operation, which are unary, must be given.
Since the set of variables are disjoint, it is impossible for name capture to happen. 
Besides, since substitution does not perform any renaming, the identity of the $\lambda$-terms is preserved during the operation. 
As a result, $\alpha$-conversion is seldom used in the whole development.
However, since the syntax allows to build $\lambda$-term that do not have an ordinary interpretation, i.e., those who mention names from the set of the bound variables which are actually not bound to any $\lambda$-abstraction or $\Pi$-type, 
a wellformedness predicate must accompany many results to rule them out from consideration and which becomes rather ubiquitous. 

Barras and Werner \cite{barras96a,barras96aShort} mechanized a substantial part of the metatheory of the Calculus of Constructions (CC) in Coq culminating with the decidability of the type system. The syntax uses de Bruijn notation. 
In an unpublished work \cite{barraspts}, Barras extended the metatheoretic results to the PTS and CTS. The sources can be found in \cite{urbanSources}.

In \cite{urban2011}, Urban et al. used the Isabelle/HOL system \cite{nipkow2002}   together with the Nominal Datatype Package (NDP) or Nominal Isabelle \cite{ndp} to mechanize the metatheory of LF \cite{harper2005}.
The NDP provides a framework to work with inductive types with binders and their associated induction and recursion principles modulo $\alpha$-conversion. 
A limitation of their approach is that the NDP does not yet allow generating executable code, so an implementation for a type-checker cannot be extracted directly. 
Also, since HOL is founded on classical logic, the results about the decidability of LF were not entirely formalized. The sources for LF and NDP can be found respectively in \cite{ptsSourcesGit} and in Isabelle's distribution \cite{isabelle2016}.

Next we point out the reader some work at the forefront in the mechanization of type theory; 
these developments focus on larger object theories, e.g., featuring universes, large elimination, $\Sigma$-types, and so on, thus a comparison with our work does not seem relevant yet (all of them use de Bruijn indices): 
In \cite{abel2018}, Abel et al. present a proof of the decidability of conversion for a fragment of Martin-L\"of's Type Theory in Agda;  
Adjedj et al. \cite{adjedj24} use Coq to mechanize a proof of the decidability for a type system of similar characteristics to the one above, and; in \cite{sozeau2025}, Sozeau et al. present a partial formalization about the decidability of a considerable part of Coq's kernel (normalization is assumed), written in Coq.

\section{Conclusions}
\label{sec:conclusions}

We have formalized on machine some interesting body of knowledge for the PTS using conventional syntax, i.e., first-order with one-sorted variables names and without identifying $\alpha$-convertible $\lambda$-terms. Among other results we have proven: weakening, syntactic validity, closure under $\alpha$-conversion and substitution.
In the course we had to update the existing framework of Stoughton's substitutions with Church-style $\lambda$-abstractions and $\Pi$-types.
We have also given a new definition for $\alpha$-conversion that works better than previous ones in the sense that it enabled us to prove some key lemmas by using simpler methods.
Except for the soundness and completeness of this new definition, which is not used in the whole development,
all results follow by structural induction on the various relations defined. 

The use of conventional syntax allowed us to prove closure of $\beta$-conversion under substitution directly by using structural induction. In contrast, McKinna and Pollack had to give an alternative definition using infinitary rules. The proof that both characterizations define the same relation follows by a (supplementary) multiple renaming lemma.
This is also the case for their equivalence result relating the different presentations of the PTS. 
In our case, we did not have to prove any equivalence result for $\beta$-conversion since our single definition is already standard.
As to the equivalence of the type systems, unlike them, it bears repeating, we have been able to reuse the substitution lemma.

The resulting Agda code is within the limits of what is manageable. 
The entire development spans about 3000LOC and it is divided equally between the framework and the metatheory of the PTS.
To put in perspective, the work by Urban et al. is about 15KLOC, from which 1800LOC belongs to the metatheory up to syntactic validity, a bit larger than our counterpart.%\footnote{We have used only the file LF.thy located in \cite{urban_sources} to count the lines of code related to the metatheoretic results mentioned.}
The NDP on the other hand is over 9300LOC.%\footnote{Counting the lines in the theory file Nominal.thy and in the ML plugins inside the Nominal library shipped with the Isabelle2016-1 distribution \cite{isabelle2016}.}
As to the work by Barras on the PTS in \cite{barraspts} (which does not address the problem of names at all since it uses de Bruijn indices), the entire corresponding formalization is approximately 2600LOC.%\footnote{Using the following files in \cite{pts_sources_git}: Alambda.v, Arules.v, Atermes.v, Atyping.v, Beta.v, Env.v, General.v, Lambda_Rules.v, Llambda.v, Lrules.v, Ltermes.v, Ltyping.v, Metatheory.v, MyList.v, MyRelations.v, PTS_spec.v, Rules.v and Termes.v.}. 

\bibliography{references}

\end{document}